\documentclass[journal]{article}

\usepackage{graphicx,url}
\usepackage{dcolumn}
\usepackage{bm}
\usepackage{amsmath,amsfonts,amssymb,citesort}

\newtheorem{definition}{\noindent{\it Definition}}[section]
\newtheorem{theorem}{\noindent{\it Theorem}}[section]

\newtheorem{lemma}[theorem]{\noindent{\it Lemma}}

\newtheorem{remark}[theorem]{\noindent{\it Remark}}

\newenvironment{proof}{\noindent{\it Proof:}}{$\hfill$ $\Box$\\ }

\begin{document}

\title{On optimal constacyclic codes}

\author{Giuliano G. La Guardia
\thanks{Giuliano Gadioli La Guardia is with Department of Mathematics and Statistics,
State University of Ponta Grossa (UEPG), 84030-900, Ponta Grossa,
PR, Brazil. }}

\maketitle

\begin{abstract}
In this paper we construct new families of
maximum-distance-separable (MDS) convolutional codes derived from
classical constacyclic codes. Additionally, we also construct new
families of almost MDS block and convolutional codes. Most of the
new convolutional codes constructed here are unit-memory and have
degree $\gamma=2$. Moreover, based on a different method, we
construct families of asymmetric quantum codes known in the
literature. All these results are mainly derived from the
investigation of suitable properties on cyclotomic cosets of such
codes.
\end{abstract}

\section{Introduction}\label{Intro}

The class of constacyclic codes
\cite{Berlekamp:1967,Krishna:1990,Aydin:2001,Blackford:2008,Liu:2012,Bakshi:2012}
has been investigated in literature. This class of codes contains
the well-known class of cyclic and negacyclic codes
\cite{Berlekamp:1967,Macwilliams:1977,Huffman:2003}. Although the
theory of cyclic codes is more investigated in literature, one can
also obtain codes with good or even optimal parameters when
considering constacyclic codes not equivalent to the cyclic ones.
The structure of constacyclic codes is in a certain sense similar to
cyclic codes, since they are principal ideals in the corresponding
quotient rings.

On the other hand, much investigation has been done concerning the
theory of convolutional codes with their corresponding properties,
as well as constructions of codes with good parameters or even
maximum-distance-separable (MDS) codes (i.e., codes attaining the
generalized Singleton bound \cite[Theorem 2.2]{Rosenthal:1999})
\cite{Forney:1970,Lee:1976,Piret:1988,York:1999,Rosenthal:1999,Rosenthal:2001,Gluesing:2006,LaGuardia:2013I,LaGuardia:2013II,LaGuardia:2013IV,LaGuardia:2016}.

In the last two decades, many researches have focussed the attention
in the investigation of properties of asymmetric quantum
error-correcting codes (AQECC)
\cite{Steane:1996,Calderbank:1998,Steane:1999,Ioffe:2007,Sarvepalli:2009,LaGuardia:2013III},
as well as constructions of AQECC with good or optimal parameters
(optimal in the sense that the code parameters attain the quantum
version of Singleton bound \cite[Lemma 3.2]{Sarvepalli:2009} with
equality; these codes are called maximum-distance-separable or MDS).

It is well-known that searching for codes with good (or optimal)
parameters always received much attention in literature. In order to
contribute to this topic of research, we propose constructions of
new families of optimal convolutional codes as well as new families
of optimal asymmetric quantum codes. To be more specific, in this
paper we utilize the families of classical constacyclic MDS codes
constructed here in order to obtain optimal convolutional codes.
Further, we construct two families of almost MDS codes of length
$n=2q+2$ over ${\mathbb F}_{q}$, after deriving the corresponding
almost MDS convolutional codes. Additionally, we also utilize these
classical MDS codes constructed in this paper to derive families of
asymmetric quantum MDS codes. To construct our MDS constacyclic
codes, we will compute and characterize all cyclotomic cosets of
these codes.

This paper is arranged as follows. In Section~\ref{secII} we give
the necessary background for the development of this work:
Section~\ref{secIIA} establishes known results on constacyclic codes
and Section~\ref{secIIB} presents a brief review of convolutional
codes. Section~\ref{secIII} presents constructions of classical MDS
constacyclic codes over ${\mathbb F}_{q}$ of length $n=q+1$ and
$n=(q+1)/2$, where $q$ is a prime power; see Lemmas~\ref{cyclo1} to
\ref{cyclo4}). These codes are obtained mainly by investigating the
structure of the corresponding cosets. In Section~\ref{secIV} we
construct new families of MDS and almost MDS convolutional codes. As
mentioned in the Abstract, most of the new convolutional codes
constructed here are unit-memory and have degree $\gamma=2$. In
Section~\ref{secV}, we propose constructions of new families of
optimal asymmetric quantum codes and finally, in
Section~\ref{secVI}, we give a brief summary of this paper.

\section{Preliminaries}\label{secII}

In this section we present preliminary results for the development
of this paper. We begin with a review of constacyclic codes, after
giving a brief review of convolutional codes.

\subsection{Constacyclic codes}\label{secIIA}

Throughout this paper, we always assume that $q$ is a prime power,
${\mathbb F}_{q}$ is a finite field with $q$ elements, ${\mathbb
F}_{q}^{*}= {\mathbb F}_{q}-\{ 0\}$ is the (cyclic) multiplicative
group of ${\mathbb F}_{q}$. The order of an element $\alpha \in
{\mathbb F}_{q}^{*}$ denoted by ${\operatorname{ord}}{(\alpha)}$, is
the smallest positive integer $t$ such that ${\alpha}^{t}=1$. We
always consider that $n$ (the code length) is a positive integer
with $\gcd (n, q)=1$. As usual, we utilize the notation $[n, k,
d]_{q}$ to denote the parameters of a classical linear code of
length $n$, dimension $k$ and minimum distance $d$.

Let $r$ be a positive integer with $r|(q-1)$, $rk=q-1$, and consider
that $\beta$ is a primitive $rn$-th root of unity in ${\mathbb
F}_{q^{m}}$, where $m= \ {\operatorname{ord}_{rn}}(q)$ denotes the
multiplicative order of $q$ (mod $rn$), (i.e, $m$ is the smallest
positive integer such that $rn | (q^{m} -1$)). Thus $\xi =
{\beta}^{r} \in {\mathbb F}_{q^{m}}$ is a primitive $n$-th root of
unity. Moreover $[{\beta}^{n}]^{r}=1\Longrightarrow
[{\beta}^{nr}]^{k}=1$ $\Longrightarrow
[{\beta}^{n}]^{(q-1)}=1\Longrightarrow
[{\beta}^{n}]^{q}={\beta}^{n}$. Hence, it follows that
${\beta}^{n}\in {\mathbb F}_{q}$. Since $\alpha={\beta}^{n}\neq 0$
then $\alpha\in {\mathbb F}_{q}^{*}$ is an element of order $r$.
Thus, if we consider the quotient ring $R_{n}={\mathbb F}_{q}[x]/(
x^{n}- \alpha )$, an $\alpha$-constacyclic code $C$ of length $n$ is
a principal ideal of $R_{n}$ under the usual correspondence ${\bf c}
= (c_0 , c_1 , \ldots , c_{n-1})\longrightarrow c_0 + c_1 x + \ldots
+ c_{n-1} x^{n-1}$ (mod$(x^{n}-\alpha)$). The generator polynomial
$g(x)$ of $C$ satisfies $g(x)|(x^{n}-\alpha)$. The roots of
$x^{n}-\alpha$ are given by ${\beta}^{1+ri}$ for all $0 \leq i \leq
n-1$. Let us consider the set ${\mathbb O}_{rn}= \{1+ri, 0\leq i\leq
n-1\}$; the defining set of $C$ is given by ${ \mathcal Z} = \{j \in
{\mathbb O}_{rn}| {\beta}^{j} \ \ is \ \ root \ \ of \ \ g(x)\}$.
The defining set is a union of some $q$-ary cyclotomic cosets modulo
$rn$ given by ${\mathbb C}_{s} = \{s, sq, \ldots, sq^{(m_{s}-1)}
\}$, where $m_{s}$ is the smallest positive integer such that
$sq^{(m_{s})}\equiv s$ (mod $rn$). The minimal polynomial (over
${\mathbb F}_{q}$) of ${\beta}^{j} \in { \mathbb F}_{q^{m}}$ is
denoted by ${M}^{(j)}(x)$ and it is given by
${M}^{(j)}(x)=\displaystyle \prod_{j \in {\mathbb{C}}_{i}}
(x-{\beta}^{j})$. The dimension of $C$ is given by $n - |{ \mathcal
Z}|$. Let us recall the concept of $\alpha$-constacyclic BCH codes:

\begin{definition}(Constacyclic BCH codes)\label{constadefi}
Let $q$ be a prime power with $\gcd (n, q)=1$. Let $\beta$ be a
primitive $rn$-th root of unity in ${\mathbb F}_{q^{m}}$. A
$\alpha$-constacyclic code $C$ of length $n$ over ${\mathbb F}_{q}$
is a BCH code with designed distance $\delta$ if, for some $b=1+ri$
we have $g(x)= \operatorname{lcm} \{{M}^{(b)}(x), {M}^{(b+r)}(x),
\ldots, {M}^{[b+r(\delta-2)]}(x)\}$, i.e., $g(x)$ is the monic
polynomial of smallest degree over ${\mathbb F}_{q}$ having
${{\alpha}^{b}}, {{\alpha}^{b+r}}\ldots,$
${{\alpha}^{[b+r(\delta-2)]}}$ as zeros.
\end{definition}

From Definition~\ref{constadefi}, it follows that $c\in C$ if and
only if $c({\alpha}^{b})=c({{\alpha}^{(b+r)}})=\ldots =
c({{\alpha}^{[b+r(\delta-2)]}})=0$. Thus the code has a string of
$\delta - 1$ consecutive $r$ powers of $\beta$ as zeros.

The BCH bound for Constacyclic codes (see
\cite{Krishna:1990,Aydin:2001} for instance) establishes that if $C$
is an $\alpha$-constacyclic code of length $n$ over ${\mathbb
F}_{q}$ with generator polynomial $g(x)$ and if $g(x)$ has the
elements $\{ {\beta}^{1+ri}| 0 \leq i\leq d-2 \}$ as roots, ($\beta$
is a primitive $rn$-th root of unity), then the minimum distance
$d_{C}$ of $C$ satisfies $d_{C}\geq d$. Recall that given an
arbitrary $[n, k, d]_{q}$ linear code $C$, the Singleton bound
asserts that $d\leq n-k+1$. If the parameters of $C$ satisfy $d =
n-k+1$, the code is called maximum-distance-separable (MDS) or
optimal.

\subsection{Convolutional Codes}\label{secIIB}

The class of convolutional codes is a well-studied class of codes
\cite{Forney:1970,Piret:1988,Johannesson:1999,Huffman:2003,Aly:2007}.
We assume the reader is familiar with the theory of convolutional
codes. For more details, see \cite{Johannesson:1999}. Recall that a
polynomial encoder matrix $G(D)=(g_{ij}) \in {\mathbb F}_{q}{[D]}^{k
\times n}$ is called \emph{basic} if $G(D)$ has a polynomial right
inverse. A basic generator matrix is called \emph{reduced} (or
minimal \cite{Johannesson:1999,Huffman:2003}) if the overall
constraint length $\gamma =\displaystyle\sum_{i=1}^{k} {\gamma}_i$,
where ${\gamma}_i = {\max}_{1\leq j \leq n} \{ \deg g_{ij} \}$, has
the smallest value among all basic generator matrices (in this case
the overall constraint length $\gamma$ will be called the
\emph{degree} of the resulting code).

\begin{definition}\cite{Johannesson:1999}
A rate $k/n$ \emph{convolutional code} $C$ with parameters $(n, k,
\gamma ; m,$ $d_{f} {)}_{q}$ is a submodule of ${ \mathbb F}_{q}
{[D]}^{n}$ generated by a reduced basic matrix $G(D)=(g_{ij}) \in {
\mathbb F}_q {[D]}^{k \times n}$, that is, $C = \{ {\bf u}(D)G(D) |
{\bf u}(D)\in {\mathbb F}_{q} {[D]}^{k} \}$, where $n$ is the
length, $k$ is the dimension, $\gamma =\displaystyle\sum_{i=1}^{k}
{\gamma}_i$ is the \emph{degree}, where ${\gamma}_i = {\max}_{1\leq
j \leq n} \{ \deg g_{ij} \}$, $m = {\max}_{1\leq i\leq
k}\{{\gamma}_i\}$ is the \emph{memory} and $d_{f}=$wt$(C)=\min
\{wt({\bf v}(D)) \mid {\bf v}(D) \in C, {\bf v}(D)\neq 0 \}$ is the
\emph{free distance} of the code.
\end{definition}

%

Recall that if ${[n, k, d]}_{q}$ denotes the parameters of a block
code with parity check matrix $H$, then we split $H$ into $m+1$
disjoint submatrices $H_i$ such that
\begin{eqnarray}\label{1}
H = \left[
\begin{array}{c}
H_0\\
H_1\\
\vdots\\
H_{m}\\
\end{array}
\right], \end{eqnarray}
where each $H_i$ has $n$ columns, obtaining
the polynomial matrix
\begin{eqnarray}\label{2}
G(D) =  {\tilde H}_0 + {\tilde H}_1 D + {\tilde H}_2 D^2 + \ldots +
{\tilde H}_m D^m,
\end{eqnarray}
having $\kappa$ rows, where $\kappa$ is the maximal number of rows
among the matrices $H_i$. The matrices ${\tilde H}_i$ for all $1\leq
i\leq m$, are derived from the respective matrices $H_i$ by adding
zero-rows at the bottom in such a way that the matrix ${\tilde H}_i$
has $\kappa$ rows in total. The matrix $G(D)$ generates a
convolutional code.

\begin{theorem}\cite[Theorem 3]{Aly:2007}\label{A}
Let $\operatorname{rk}$ denote the rank of a matrix. Let $C
\subseteq { \mathbb F}_{q}^{n}$ be a linear code with parameters
${[n, k, d]}_{q}$ and assume also that $H \in { \mathbb
F}_q^{(n-k)\times n}$ is a parity check matrix of $C$ partitioned
into submatrices $H_0, H_1, \ldots, H_m$ as in Eq.~(\ref{1}) such
that $\kappa = \operatorname{rk} H_0$ and $\operatorname{rk} H_i
\leq \kappa$ for $1 \leq
i\leq m$. Consider the polynomial matrix $G(D)$ as in Eq.~(\ref{2}). Then we have:\\
(a) The matrix $G(D)$ is a reduced basic generator matrix;\\
(b) Let $V$ be the convolutional code generated by $G(D)$ and
$V^{\perp}$ its Euclidean dual code. If $d_f$ and $d_f^{\perp}$
denote the free distances of $V$ and $V^{\perp}$, respectively,
$d_i$ denote the minimum distance of the code $C_i = \{ {\bf v}\in {
\mathbb F}_{q}^{n} \mid {\bf v} {\tilde H}_i^t =0 \}$ and
$d^{\perp}$ is the minimum distance of $C^{\perp}$, then one has
$\min \{ d_0 + d_m , d \} \leq d_f^{\perp} \leq  d$ and $d_f \geq
d^{\perp}$.
\end{theorem}


\section{Classical MDS Codes}\label{secIII}

In this section we propose a construction of optimal
$\alpha$-constacyclic codes. It is well known that an $[n, k,
d{]}_{q}$ $\alpha$-constacyclic MDS code of length $n=q+1$ over ${
\mathbb F}_{q}$ exists for $k$ odd if $\alpha$ is a quadratic
residue in ${ \mathbb F}_{q}$ and for $k$ even if $\alpha$ is not a
quadratic residue in ${ \mathbb F}_{q}$ (see \cite{Krishna:1990}).
However, the constructions presented here differ from the ones known
in the literature since they consider properties of $q$-ary
cyclotomic cosets ($q$-cosets) modulo $rn$ in ${\mathbb O}_{rn}$. In
order to proceed further, we will show some useful lemmas that will
be utilized in our constructions. Lemmas~\ref{cyclo1} to
\ref{cyclo4} presented here are generalizations of Lemmas 4.1 and
4.4 in \cite{Kai:2013}.

\begin{lemma}\label{cyclo1}
Let $q$ be a power of an odd prime and $n=q+1$. Assume that $\alpha
\in { \mathbb F}_{q}^{*}$ is an element of order $r\geq 2$, where
$q-1=rk$, $k$ even. Put $s=n/2$. Then the following hold:
\begin{itemize}
\item [ a)] The unique $q$-cosets modulo $rn$, in ${\mathbb
O}_{rn}$, with only one element are the cosets ${\mathbb C}_{s}=\{
s\}$ and ${\mathbb C}_{[(r+1)s]}=\{ (r+1)s\}$.

\item [ b)] The remaining $q$-cosets modulo $rn$, in ${\mathbb O}_{rn}$,
are given by ${\mathbb
C}_{(s-ri)}=\{ s-ri, s+ri\}$, where $1\leq i\leq s-1$.
\end{itemize}
\end{lemma}

\begin{proof}
Notice that $\gcd (q, n) =1$ and ${\operatorname{ord}}_{rn}(q) =
2$ are true.\\
a) Let us consider $c=k/2$. Then we have
$sq=\frac{n}{2}(q-1)+\frac{n}{2}=\frac{nrk}{2}+\frac{n}{2}=
cnr+\frac{n}{2}\equiv s$ (mod $rn$). On the other hand, $(r+1)sq =
rsq + sq\equiv rs + s=(r+1)s$ (mod $rn$).\\
b) We have: $(s+ri)q\equiv s+rqi= s+r(q+1)i -ri\equiv s-ri$ (mod
$rn$). Since $s= 1+r\frac{k}{2}$, it follows that $s \in {\mathbb
O}_{rn}$. Therefore, the elements $s-ri, s+ri \in {\mathbb O}_{rn}$
for all $1\leq i\leq s-1$. Moreover, it is easy to see that these
cosets are mutually disjoint. Further, with exception of the cosets
${\mathbb C}_{s}=\{ s\}$ and ${\mathbb C}_{[(r+1)s]}=\{ (r+1)s\}$,
all the remaining cosets have two elements. The union of the cosets
have $2(s-1)+2 = n$ elements, hence the result follows.
\end{proof}

\begin{lemma}\label{cyclo2}
Let $q$ be a power of an odd prime and $n=q+1$. Assume that $\alpha
\in { \mathbb F}_{q}^{*}$ is an element of order $r\geq 2$, where
$q-1=rk$, $k$ odd. Then the following hold:
\begin{itemize}
\item [ a)] If $t=(n+r)/2$ then the $q$-coset in ${\mathbb O}_{rn}$ containing
$t$ is of the form ${\mathbb C}_{t}=\{t, t-r\}$.

\item [ b)] The remaining $q$-cosets in ${\mathbb O}_{rn}$ are of the form ${\mathbb
C}_{(t+ri)}=\{t+ri, t-ri-r \}$, where $1\leq i\leq n/2 -1$.
\end{itemize}
\end{lemma}

\begin{proof}
Note that $t=(n+r)/2= 1+r\left[\frac{(k+1)}{2}\right]$, so the
element $t=(n+r)/2$ is of the form $1+ri$. Thus the elements $t+ri,
t-ri-r$ also belong
to ${\mathbb O}_{rn}$ for all $1\leq i\leq n/2 - 1$.\\
a) Consider that $c=(k+1)/2$. Then one has: $tq=\frac{(n+r)}{2}q=
\frac{[r(k+1)+2]}{2}q=(rc+1)q= (rc+1)(q+1)-rc-1\equiv q-rc=
q-r\frac{(k+1)}{2}= \frac{[2q-(q-1)-r]}{2}=
\frac{(q+1-r)}{2}=\frac{n-r}{2}=
\frac{n+r}{2}-r = t-r$ (mod $rn$).\\
b) We have: $(t+ri)q\equiv t-r+rqi= t-r+ri(q+1)-ri\equiv t-r-ri$
(mod $rn$). Similarly to the proof of Lemma~\ref{cyclo1}, it is easy
to see that these cosets are mutually disjoint and the union of them
has $n$ elements.
\end{proof}

\begin{lemma}\label{cyclo3}
Let $q\equiv 1$ $\operatorname{(} \operatorname{mod} 4
\operatorname{)}$ be a power of an odd prime, $n=(q+1)/2$ and
$\alpha \in { \mathbb F}_{q}^{*}$ such that
${\operatorname{ord}}{(\alpha)}= r\geq 2$, where $q-1=rk$, $k$ even.
Then the following hold:
\begin{itemize}
\item [ a)] One has ${\mathbb C}_{n}=\{n\}$ modulo $rn$.

\item [ b)] The remaining $q$-cosets in ${\mathbb O}_{rn}$ are of the form ${\mathbb
C}_{(n-ri)}=\{n-ri, n+ri \}$ $\operatorname{(} \operatorname{mod} rn
\operatorname{)}$, where $1\leq i\leq (n-1)/2$.
\end{itemize}
\end{lemma}

\begin{proof}
Since $k$ is even then $n= 1 + r\frac{k}{2} \in {\mathbb O}_{rn}$.
Items a) and b) follow by straightforward computation (note that,
from hypotheses, $n$ is odd).
\end{proof}

\begin{lemma}\label{cyclo4}
Let $q\equiv 3$ $\operatorname{(} \operatorname{mod} 4
\operatorname{)}$ be a power of an odd prime, $n=(q+1)/2$ and
$\alpha \in { \mathbb F}_{q}^{*}$ such that
${\operatorname{ord}}{(\alpha)}= r\geq 2$, where $q-1=rk$, $k$ even.
Then the following hold:
\begin{itemize}
\item [ a)] The unitary $q$-cosets modulo $rn$, in ${\mathbb O}_{rn}$, are given by
${\mathbb C}_{n}=\{n\}$ and ${\mathbb
C}_{\frac{(r+2)n}{2}}=\left\{\frac{(r+2)n}{2} \right\}$.

\item [ b)] The remaining $q$-cosets modulo $rn$, in ${\mathbb O}_{rn}$, are of the form
${\mathbb C}_{(n-ri)}=\{n-ri, n+ri \}$, where $1\leq i\leq n/2-1$.
\end{itemize}
\end{lemma}

\begin{proof}
Note that in this case $n$ is even.
\end{proof}

\begin{remark}
It is interesting to observe that in Lemmas~\ref{cyclo1} to
\ref{cyclo4} we have assumed $r\geq 2$ to avoid the case when
$\alpha =1$, i.e., cyclic codes.
\end{remark}

\begin{lemma}\label{cyclo5}
Let $q= 2^{t}$, where $t\geq 2$, and $n=q+1$. Assume that $\alpha
\in { \mathbb F}_{q}^{*}$ is such that
${\operatorname{ord}}{(\alpha)}= r$, where $q-1=rk$. Suppose also
that $i_{0}=\frac{(k-1)}{2}$, $s = 1 + ri_{0}$ and $t= s +
r(1+\frac{q}{2})$ $\operatorname{(}t$ is considered
$\operatorname{(} \operatorname{mod} rn
\operatorname{)}\operatorname{)}$. Then the following hold:
\begin{itemize}
\item [ a)] ${ \mathbb C}_{s}=\{s, s+r\}$ and ${\mathbb C}_{t}=\{ t\}$
in modulo $rn$, in ${\mathbb O}_{rn}$.

\item [ b)] The remaining $q$-ary cosets modulo $rn$, in ${\mathbb O}_{rn}$, are
given by ${\mathbb C}_{(s-ri)}=\{s-ri, s+r+ri \}$, where $1\leq
i\leq q/2-1$.
\end{itemize}
\end{lemma}

\begin{proof}
We only show Item a): $sq=\left[ 1+ r\frac{(k-1)}{2} \right] q = q +
r(q+1)\frac{(k-1)}{2} -r\frac{(k-1)}{2} \equiv rk + 1
-r\frac{(k-1)}{2} = 1 +r\frac{(k-1)}{2} + r$ ( mod $rn$). On the
other hand, $tq \equiv s + r + rq + rq\frac{q}{2} = s + r + rq +
r(q+1)\frac{q}{2} - r\frac{q}{2} \equiv s + r + r\frac{q}{2}=t$ (
mod $rn$).
\end{proof}

%
%
%

Until now we have dealt with properties and characterizations of
cyclotomic cosets. In the following theorems, we utilize the
previous results about cosets for the construction of optimal
constacyclic codes. As we will see in the following sections, these
codes will be utilized in order to construct families of optimal
convolutional codes and also to derive families of optimal
asymmetric quantum codes.

\begin{theorem}\label{mainclasI}
Assume that all hypotheses of Lemma~\ref{cyclo1} hold. Then there
exist MDS constacyclic codes with parameters $[n, n-2i-1,
2i+2{]}_{q}$, for every $0\leq i\leq \frac{n}{2} -1$.
\end{theorem}
\begin{proof}
We adopt the same notation of Lemma~\ref{cyclo1}. Let us consider
the $\alpha$-constacyclic code with defining set ${ \mathcal Z} =
\bigcup_{l=0}^{i}{\mathbb C}_{(s+rl)}$, where  ${\mathbb C}_{s}=\{
s\}$ and ${\mathbb C}_{(s-ri)}=\{ s-ri, s+ri\}$, for every $0\leq i
\leq n/2 -1$. The defining set of code $C$ consists of exactly $i$
cosets with two elements and one coset having only one element (see
Lemma~\ref{cyclo1}), i.e., ${ \mathcal Z}$ has $2i + 1$ elements; so
$C$ has dimension $k=n-2i-1$. From the BCH bound for constacyclic
codes and by construction, $C$ has minimum distance $d \geq 2i+2$.
By applying the Singleton bound, it follows that $d=2i+2$, hence
there exists an $[n, n-2i-1, 2i+2{]}_{q}$ MDS code, for every $0\leq
i\leq n/2 -1$.
\end{proof}

\begin{theorem}\label{mainclasII}
Assume that all hypotheses of Lemma~\ref{cyclo2} hold. Then there
exist MDS constacyclic codes with parameters $[n, n-2i-2,
2i+3{]}_{q}$, for every $0\leq i \leq \frac{n}{2} -2$.
\end{theorem}
\begin{proof}
Let $C$ be the $\alpha$-constacyclic code with defining set ${
\mathcal Z} = \bigcup_{l=0}^{i}{\mathbb C}_{(t+rl)}$, $0\leq i \leq
n/2 - 2$. Analogously to the previous proof, the defining set of
code $C$ consists of exactly $i+1$ $q$-cosets. From
Lemma~\ref{cyclo2}, each of them has two elements, i.e., ${ \mathcal
Z}$ has $2(i + 1)$ elements; so $C$ has dimension $k= n-2i-2$. From
the BCH bound, from construction and by applying the Singleton
bound, one concludes that $C$ is an $[n, n-2i-2, 2i+3{]}_{q}$ MDS
code, for every $0\leq i\leq n/2 -2$.
\end{proof}

\begin{theorem}\label{mainclasIII}
Assume that all hypotheses of Lemma~\ref{cyclo3} hold. Then there
exist MDS constacyclic codes with parameters $[n, n-2i-1,
2i+2{]}_{q}$, for every $0\leq i\leq \frac{(n-1)}{2}-1$.
\end{theorem}
\begin{proof}
Let $C$ be the $\alpha$-constacyclic code with defining set ${
\mathcal Z} = \bigcup_{l=0}^{i}{\mathbb C}_{(n-rl)}$. The defining
set of code $C$ consists of exactly $i$ cosets with two elements and
one coset having only one element (see Lemma~\ref{cyclo3}), i.e., ${
\mathcal Z}$ has $2i + 1$ elements; so $C$ has dimension $k=n-2i-1$.
From the BCH bound, from construction and by applying the Singleton
bound, one concludes that $C$ is an $[n, n-2i-1, 2i+2{]}_{q}$ MDS
code.
\end{proof}

\begin{remark}\label{obs1}
Applying Lemma~\ref{cyclo4} one can get analogous results to that of
Theorem~\ref{mainclasIII}.
\end{remark}

\begin{theorem}\label{mainclasIIIA}
Assume all hypotheses of Lemma~\ref{cyclo5} hold. Then there exist
MDS constacyclic codes with parameters $[n, n-2i-2, 2i+3{]}_{q}$,
for every $0\leq i\leq \frac{(n-1)}{2}-2$.
\end{theorem}
\begin{proof}
Analogous to that of Theorem~\ref{mainclasI}.
\end{proof}

\begin{theorem}\label{mainclasIIIB}
Assume all hypotheses of Lemma~\ref{cyclo5} hold. Then there exist
MDS constacyclic codes with parameters $[n, n-2i-1, 2i+2{]}_{q}$,
for every $0\leq i\leq \frac{(n-1)}{2}-1$.
\end{theorem}
\begin{proof}
Similar to that of Theorem~\ref{mainclasI}.
\end{proof}

Recall the an $[n, k, d]_{q}$ linear code $C$ is called
\emph{almost} MDS if it has Singleton defect $1$, i.e., $k=n-r-1$
and $d=r+1$. In this context, we construct a simple family of almost
MDS codes in the following theorem:

\begin{theorem}\label{mainclasIV}
Assume that $q\equiv 3 \operatorname{(mod} 4\operatorname{)}$ is a
power of an odd prime and let $n=2q+2$, Then there exist almost MDS
codes with parameters $[n, n-4, d\geq 4]_{q}$.
\end{theorem}
\begin{proof}
It is easy to see that $\gcd (q, n)=1$ and $\
{\operatorname{ord}}_{n}(q) = 2$. We will construct a cyclic (
$\alpha =1$) code having these parameters. Putting $s=(q-1)/2$, we
will verify that the $q$-ary coset containing $s$ is given by
${\mathbb C}_{s}= \{s, s+2\}$. In fact, since $q=4k+3$ for some $k
\in {\mathbb N}$, we have $[\frac{q-1}{2}]q = 8k^{2}+10k+3$. Because
$8k = n-8$, it follows that $sq\equiv 2k+3 = s+2$ (mod $n$).
Further, the coset ${\mathbb C}_{(s+1)}$ has two elements and
${\mathbb C}_{s}$ and ${\mathbb C}_{(s+1)}$ are disjoint. Then the
BCH code with defining set ${\mathcal Z}={\mathbb C}_{s}\cup{\mathbb
C}_{(s+1)}$ has parameters $[n, n-4, d\geq 4]_{q}$, as desired.
\end{proof}

\begin{theorem}\label{mainclasIVA}
Let $q$ be a power of an odd prime and consider that $n=2q-2$.\\
a) If $q \geq 5$ then there exist almost MDS codes with parameters
$[n, n-4, d\geq 4]_{q}$.\\
b) If $q \geq 7$ then there exist codes with parameters $[n, n-7,
d\geq 6]_{q}$.
\end{theorem}
\begin{proof}
Note that the $q$-cosets ${\mathbb C}_{s}$ (mod $n$) for $s$ even
are singletons.\\
a) The cosets ${\mathbb C}_{3}$ has two elements. The BCH code with
defining set ${\mathcal Z}=\bigcup_{i=2}^{4}{\mathbb C}_{i}$ has the
specified parameters.\\
b) Since the cosets ${\mathbb C}_{1}$ and ${\mathbb C}_{3}$ are
disjoint and have two elements, the result follows by taking the BCH
code with defining set ${\mathcal Z}=\bigcup_{i=0}^{4}{\mathbb
C}_{i}$.
\end{proof}
\begin{remark}
Note that we have assumed that $q\geq 7$ in Item b) of
Theorem~\ref{mainclasIVA} because for $q=5$ one can get an $[8, 1,
d\geq 7]_{5}$ code and not an $[8, 1, d\geq 6]_{5}$ as stated in
such result.
\end{remark}

\begin{theorem}\label{mainclasV}
Assume that $q\equiv 1 \operatorname{(mod} 4 \operatorname{)}$ is a
power of an odd prime and let $n=2q+2$, Then there exist almost MDS
codes with parameters $[n, n-3, d\geq 3]_{q}$.
\end{theorem}
\begin{proof}
It is easy to see that the $q$-ary coset ${\mathbb
C}_{\frac{(q+1)}{2}}$ is given by ${\mathbb
C}_{\frac{(q+1)}{2}}=\left\{\frac{(q+1)}{2}\right\}$ and the coset
${\mathbb C}_{\left[\frac{(q+1)}{2}+1\right]}$ have two elements.
Choosing the BCH (cyclic) code with defining set ${\mathcal
Z}={\mathbb C}_{\frac{(q+1)}{2}}\cup{\mathbb
C}_{\left[\frac{(q+1)}{2}+1\right]}$, the result follows.
\end{proof}

\begin{remark}
Let us recall the well-known MDS Conjecture: Let $q$ be a prime
power. If there exists a nontrivial $[n, k, d]_{q}$ MDS code, then
$n \leq q+1$, except when $q$ is even and $k = 3$ or $k = q-1$, and
in this case one has $n \leq q+2$. Note that by applying
Theorems~\ref{mainclasIV}, \ref{mainclasIIIA} and \ref{mainclasV}
one can derive codes with good parameters. In fact, if the MDS
conjecture holds, those codes are the best ones known to exist.
\end{remark}

Let us discuss the results presented in this section:

\begin{itemize}
\item  The characterization of the cyclotomic cosets modulo $rn$, in
${\mathbb O}_{rn}$, is new. More precisely, the forms of the cosets
given in Lemmas $3.1$ to $3.6$ are new.

\item  Although the method of construction of MDS constacyclic codes
by applying that cyclotomic cosets given in Lemmas $3.1$ to $3.6$ is
new, the MDS constacyclic codes derived from Lemmas $3.1$, $3.3$,
$3.4$ and $3.6$ (respectively, Theorems $3.7$, $3.9$, $3.10$, $3.11$
and $3.12$) cannot be considered as totally new, as they are
isomorphic to cyclic codes as follows. Let $\beta$  be a primitive
element in ${\mathbb F}_{q}^{*}$ such that ${\beta}^{k}=\alpha$.
When $k$ is even, the map $x\longrightarrow {\beta}^{-k/2} x$ is an
isomorphism of ${\mathbb F}_{q}[x]/( x^{q+1}-1)$ onto ${\mathbb
F}_{q} [x]/( x^{q+1}-\alpha )$, so the constacylic codes from Lemma
$3.1$ are isomorphism to cyclic codes. The codes derived from Lemmas
$3.3$ and $3.4$ are isomorphic to cyclic codes by applying the maps
$x\longrightarrow {\beta}^{t} x$, where $ t\equiv
{(\frac{q+1}{2})}^{-1} (-k)$ mod $(q-1)$ and $x\longrightarrow
{\beta}^{-k} x$, respectively. For the codes derived from Lemma
$3.6$, the map $x\longrightarrow {\beta}^{s} x$, where $s\equiv
-2^{-1}k$ mod $(q-1)$, is an isomorphism with cyclic codes.

\item  The families of MDS codes obtained from Theorem~\ref{mainclasII}
seems to be new.

\item  The families of almost MDS codes derived from
Theorems~\ref{mainclasIV}, \ref{mainclasIVA} and \ref{mainclasV} are
new.
\end{itemize}


\section{New optimal convolutional codes}\label{secIV}

Recall that the \emph{generalized Singleton bound} (see
\cite[Theorem 2.2]{Rosenthal:1999}) of an $(n, k, \gamma;$ $m, d_{f}
{)}_{q}$ convolutional code is given by $ d_{f}\leq (n-k)[ \lfloor
\gamma/k \rfloor + 1 ] + \gamma +1$. If the parameters of a
convolutional code $C$ satisfy the generalized Singleton bound with
equality then $C$ is called maximum-distance-separable (MDS) or
optimal code. In this section we shall construct new optimal
convolutional codes. To proceed further, it is necessary to
construct a parity check matrix for the (classical) constacyclic
codes considered in the construction method proposed here.

${\mathbb F}_{q}$

\begin{remark}\label{rem1}
Let ${\mathcal B} =\{ b_{1}, \ldots, b_{l}\}$ be a basis of
${\mathbb F}_{q^{l}}$ over ${\mathbb F}_{q}$. If $u = (u_1,\ldots
,u_{n}) \in {\mathbb F}_{q^{l}}^{n}$ then one can write the vectors
$u_{i}$, $1\leq i\leq n$, as linear combinations of the elements of
${\mathcal B}$, that is, $u_{i} = u_{i1}b_{1} +\ldots +
u_{il}b_{l}$. Consider that $u^{(j)} = (u_{1j},\ldots, u_{nj})$ are
vectors in ${\mathbb F}_{q}^{n}$ with $1\leq j\leq l$. Then, if $v
\in {\mathbb F}_{q}^{n}$, one has $v\cdot u=0$ if and only if $v
\cdot u^{(j)} = 0$ for all $1\leq j\leq l$.
\end{remark}

The following result is a straightforward generalization of
\cite[Theorem 5.4]{LaGuardia:2013IV}. Since we have not seen an
explicit proof of such result in literature, we present it here:

\begin{theorem}\label{paritynega}
Assume that $q$ is a prime power with $\gcd (n, q)=1$, and $m= \
{\operatorname{ord}}_{rn}(q)$. Let $\beta$ be a primitive $rn$-th
root of unity in ${\mathbb F}_{q^{m}}$ and $b=1+ri$, for some $0
\leq i\leq n-1$. Then a parity-check matrix for a BCH constacyclic
code $C$ of length $n$ and designed distance $\delta$ is given by

\begin{eqnarray*}
H_{\delta , b} =  \left[
\begin{array}{ccccc}
1 & {{\beta}^{b}} & {{\beta}^{2b}} & \cdots & {{\beta}^{(n-1)b}} \\
1 & {{\beta}^{(b+r)}} & {{\beta}^{2(b+r)}} & \cdots & {{\beta}^{(n-1)(b+r)}}\\
1 & {{\beta}^{(b+2r)}} & {{\beta}^{2(b+2r)}} & \cdots & {{\beta}^{(n-1)(b+2r)}}\\
\vdots & \vdots & \vdots & \vdots & \vdots\\
1 & {\beta}^{[b+r(\delta-2)]} & {\beta}^{2[b+r(\delta-2)]} & \cdots & {\beta}^{(n-1)[b+r(\delta-2)]}\\
\end{array}
\right],
\end{eqnarray*}
where each entry is replaced by the corresponding column of $m$
elements from ${\mathbb F}_{q}$ and then removing any linearly
dependent rows.
\end{theorem}

\begin{proof}
Assume that ${\bf c} =(c_0, c_1, \ldots , c_{n-1}) \in C$, or
equivalently, $c(x) = c_{0} + c_{1}x + \ldots + c_{n-1}x^{n-1} \in
C$. Thus we have ${\bf c}({{\beta}^{b}})={\bf
c}({{\beta}^{b+r}})={\bf c}({{\beta}^{b+2r}})= \ldots ={\bf
c}({{\beta}^{[b+r(\delta-2)]}})=0$; so it follows that

\begin{eqnarray*}
\left[
\begin{array}{ccccc}
1 & {{\beta}^{b}} & {{\beta}^{2b}} & \cdots & {{\beta}^{(n-1)b}} \\
1 & {{\beta}^{(b+r)}} & {{\beta}^{2(b+r)}} & \cdots & {{\beta}^{(n-1)(b+r)}}\\
1 & {{\beta}^{(b+2r)}} & {{\beta}^{2(b+2r)}} & \cdots & {{\beta}^{(n-1)(b+2r)}}\\
\vdots & \vdots & \vdots & \vdots & \vdots\\
1 & {\beta}^{[b+r(\delta-2)]} & {\beta}^{2[b+r(\delta-2)]} & \cdots & {\beta}^{(n-1)[b+r(\delta-2)]}\\
\end{array}
\right]\cdot \left[
\begin{array}{c}
c_0\\
c_1\\
c_2\\
\vdots\\
c_{n-1}\\
\end{array}
\right]= \left[
\begin{array}{c}
0\\
0\\
\vdots\\
0\\
\end{array}
\right]_{({\delta}-1, 1)}.
\end{eqnarray*}
From Remark~\ref{rem1} and by definition of BCH constacyclic codes,
the result follows.
\end{proof}

Now we are ready to show the main results of this section:

\begin{theorem}\label{mainI}
Assume that all hypotheses of Lemma~\ref{cyclo1} hold. Then there
exist MDS convolutional codes with parameters $(n, n-2i+1, 2 ; 1,
2i+2 {)}_{q}$, where $2\leq i \leq n/2 - 2$.
\end{theorem}

\begin{proof}
Let $C_2$ be the $\alpha$-constacyclic BCH code with parameters $[n,
n-2i-1,$ $2i+2{]}_{q}$ constructed in Theorem~\ref{mainclasI}. By
Theorem~\ref{paritynega}, a parity check matrix $H_{C_2}$ of $C_2$
is obtained from the matrix
\begin{eqnarray*}
H_{2} = \left[
\begin{array}{ccccc}
1 & {{\beta}^{s}} & {{\beta}^{2s}} & \cdots & {{\beta}^{(n-1)s}} \\
1 & {{\beta}^{(s+r)}} & {{\beta}^{2(s+r)}} & \cdots & {{\beta}^{(n-1)(s+r)}}\\
1 & {{\beta}^{(s+2r)}} & {{\beta}^{2(s+2r)}} & \cdots & {{\beta}^{(n-1)(s+2r)}}\\
\vdots & \vdots & \vdots & \vdots & \vdots\\
1 & {\beta}^{(s+ri)} & {{\beta}^{2(s+ri)}} & \cdots & {\beta}^{(n-1)(s+ri)}\\
\end{array}
\right]
\end{eqnarray*}
by expanding each entry as a column vector (containing $2$ rows)
with respect to some ${\mathbb F}_{q}-$basis $\mathcal B$ of
${\mathbb F}_{q^2}$ and then removing one linearly dependent row.

Next we assume that $C_1$ is the $\alpha$-constacyclic BCH code of
length $n$ over ${\mathbb F}_{q}$ generated by $\langle g_1
(x)\rangle =\langle {M}^{(s)}(x) {M}^{(s+r)}(x) \cdot \ldots \cdot
{M}^{[s+r(i-1)]}(x) \rangle$. Similarly, by
Theorem~\ref{paritynega}, $C_1$ has a parity check matrix $H_{C_1}$
derived from the matrix
\begin{eqnarray*}
H_{1} = \left[
\begin{array}{ccccc}
1 & {{\beta}^{s}} & {{\beta}^{2s}} & \cdots & {{\beta}^{(n-1)s}} \\
1 & {{\beta}^{(s+r)}} & {{\beta}^{2(s+r)}} & \cdots & {{\beta}^{(n-1)(s+r)}}\\
1 & {{\beta}^{(s+2r)}} & {{\beta}^{2(s+2r)}} & \cdots & {{\beta}^{(n-1)(s+2r)}}\\
\vdots & \vdots & \vdots & \vdots & \vdots\\
1 & {\beta}^{[s+r(i-1)]} & {\beta}^{2[s+r(i-1)]} & \cdots & {\beta}^{(n-1)[s+r(i-1)]}\\
\end{array}
\right]
\end{eqnarray*}
by expanding each entry as a column vector with respect to $\mathcal
B$ (already done, since $H_1$ is a submatrix of $H_{2}$) and then
removing the unique linearly dependent row. It is easy to see that
$C_1$ is an ${[n, n-2i+1, 2i]}_{q}$ MDS code.

Now, consider $C_0$ be the ${[n, n-2, d_0 \geq 2]}_{q}$
$\alpha$-constacyclic BCH code generated by ${M}^{(s+ri)}(x)$. A
parity check matrix $H_{C_0}$ of $C_0$ is given by expanding the
entries of the matrix
\begin{eqnarray*}
H_0 = \left[
\begin{array}{ccccc}
1 & {{\alpha}^{(s+2i)}} & {{\alpha}^{2(s+2i)}} & \cdots & {{\alpha}^{(n-1)(s+2i)}} \\
\end{array}
\right]
\end{eqnarray*}
with respect to $\mathcal B$ (already done, since $H_0$ is a
submatrix of $H_{2}$).

Further, let us consider that $V$ is the convolutional code
generated by the reduced basic (Theorem~\ref{A} Item (a)) generator
matrix
\begin{eqnarray*}
G(D)=\tilde H_{C_1}+ \tilde H_{C_0} D,
\end{eqnarray*}
where $\tilde H_{C_1} = H_{C_1}$ and $\tilde H_{C_0}$ is obtained
from $H_{C_0}$ by adding zero-rows at the bottom such that $\tilde
H_{C_0}$ has the number of rows of $H_{C_1}$ in total. By
construction, $V$ is a unit-memory convolutional code of dimension
$2i-1$ and degree ${\gamma}_{V} = 2$. Thus the dual $V^{\perp}$ of
the convolutional code $V$ has dimension $n-2i+1$ and degree $2$.
From Theorem~\ref{A} Item (b), the free distance of $V^{\perp}$ is
bounded by $\min \{ d_0 + d_1 , d_2 \} \leq d_{f}^{\perp} \leq d_2$.
From construction one has $d_2 = 2i+2$, $d_1 = 2i$ and $d_0 \geq 2$,
so $V^{\perp}$ has parameters $(n, n-2i+1, 2; 1, 2i+2)_{q}$.
Applying the generalized Singleton bound we conclude that
$V^{\perp}$ is MDS, as required.
\end{proof}


\begin{theorem}\label{mainII}
Assume that all hypotheses of Lemma~\ref{cyclo2} hold. Then there
exist MDS convolutional codes with parameters $(n, n-2i, 2 ; 1,
2i+3{)}_{q}$, where $2\leq i\leq n/2 - 2$.
\end{theorem}

\begin{proof}
We omit the proof since it follows the same line to that of
Theorem~\ref{mainI}.
\end{proof}

\begin{theorem}\label{mainIII}
Assume that all hypotheses of Lemma~\ref{cyclo3} hold. Then there
exist MDS convolutional codes with parameters $(n, n-2i+1, 2 ; 1,
2i+2 {)}_{q}$, where $2\leq i\leq \frac{(n-1)}{2}-1$.
\end{theorem}

\begin{proof}
Let $C_2$, $C_1$ and $C_0$ be $\alpha$-constacyclic BCH codes of
length $n$ over ${\mathbb F}_{q}$ generated, respectively, by the
product of the minimal polynomials
$$ C_2 = \langle {M}^{(n)}(x)
{M}^{(n+r)}(x) \cdot \ldots \cdot {M}^{(n+ri)}(x) \rangle,$$
$$C_1 = \langle {M}^{(n)}(x) {M}^{(n+r)}(x) \cdot \ldots \cdot
{M}^{(n+ri-r)}(x) \rangle$$ and $$C_0 =
\langle{M}^{(n+ri)}(x)\rangle,$$ where $2\leq i\leq
\frac{(n-1)}{2}-1$. Applying the same procedure shown in the proof
of Theorem~\ref{mainI}, the result follows.
\end{proof}

\begin{remark}
If we assume that all hypotheses of Lemma~\ref{cyclo4} hold then one
can obtain more optimal convolutional codes.
\end{remark}

\begin{theorem}\label{mainIIIA}
Assume all hypotheses of Theorem~\ref{mainclasIIIA} hold, with
$t\geq 3$. Then there exist MDS convolutional codes with parameters
$(n, n-2i, 2 ; 1, 2i+3 {)}_{q}$, where $2\leq i\leq
\frac{(n-1)}{2}-2$.
\end{theorem}

\begin{proof}
Similar to that of Theorem~\ref{mainI}.
\end{proof}

\begin{theorem}\label{mainIIIB}
Assume all hypotheses of Theorem~\ref{mainclasIIIB} hold, with
$t\geq 3$. Then there exist MDS convolutional codes with parameters
$(n, n-2i+1, 2 ; 1, 2i+2 {)}_{q}$, where $2\leq i\leq
\frac{(n-1)}{2}-2$.
\end{theorem}

\begin{proof}
Similar to that of Theorem~\ref{mainI}.
\end{proof}

\begin{theorem}\label{mainIV}
Let $q$ be a prime power such that $q-1=rn$ and let $\alpha \in {
\mathbb F}_{q}^{*}$ be an element of order $r\geq 1$. Then there
exist MDS convolutional codes with parameters $(n, n-c_1, c_2; 1,
i+2 {)}_{q}$, where $0\leq i\leq n-2$ and $c_1 , c_2$ are positive
integers with $i+1=c_1 + c_2$ and $c_1 \geq c_2$.
\end{theorem}
\begin{proof}
Since ${\operatorname{ord}}_{rn}(q) = 1$, then all the $q$-cosets
${\mathbb C}_{1+ri}$ (mod $rn$) are singletons. Consider $C$ to be
the $\alpha$-constacyclic code with defining set ${ \mathcal Z} =
\bigcup_{l=0}^{i}{\mathbb C}_{(1+rl)}$, where $0\leq i\leq n-2$; $C$
has parameters $[n, n-i-1, i+2]_{q}$. We next construct a
convolutional code $V$ with parameters $(n, c_1, c_2; 1,
d_{f})_{q}$. It is easy to see that its dual $V^{\perp}$ has
parameters $(n, n-c_1, c_2; 1, i+2)_{q}$ and its is a MDS code. Note
that when $r =1$, $C$ is a Reed-Solomon code over ${\mathbb F}_{q}$
(see \cite{Klappi:2007}).
\end{proof}

\begin{theorem}\label{mainV}
Assume all the hypotheses of Theorem~\ref{mainclasIV} hold. Then
there exist almost MDS convolutional codes with parameters $(n, n-2,
2 ; 1, 4{)}_{q}$.
\end{theorem}
\begin{proof}
It suffices to consider the codes constructed in
Theorem~\ref{mainclasIV}, then proceeding similarly as in the proof
of Theorem~\ref{mainI}.
\end{proof}

\begin{remark}
It is interesting to observe that the convolutional codes derived
from Theorem~\ref{mainIV} are almost MDS in the sense that their
defects are $1$ with respect to the generalized Singleton bound.
These codes seem to be new.
\end{remark}

\begin{theorem}\label{mainVI}
Assume all the hypotheses of Theorem~\ref{mainclasIVA} hold.\\
a) If $q\geq 5$, there exist almost MDS convolutional codes
with parameters $(n, n-3, 1; 1, 4{)}_{q}$.\\
b) If $q \geq 7$, there exist convolutional codes with parameters
$(n, n-4, 3; 1, d_{f}\geq 6{)}_{q}$.
\end{theorem}
\begin{proof}
Similar to that of Theorem~\ref{mainI}.
\end{proof}

\begin{remark}
Note that the new convolutional codes derived from
Theorem~\ref{mainVI} - Item b), have Singleton defect at most $2$,
hence they are also good codes.
\end{remark}

Let us now give a brief discussion of this section.

\begin{itemize}
\item  For every $q=p^{t}$, where $p$ is an odd prime and $t$ is an odd
positive integer, the families of MDS convolutional codes
constructed in Theorems~\ref{mainI}, \ref{mainII}, \ref{mainIII} are
new.

\item  The families of MDS convolutional codes constructed in
Theorems~\ref{mainIIIB} and \ref{mainIV} are new.

\item  Theorems~\ref{mainI}, \ref{mainII} and \ref{mainIII} are
generalizations of the results shown in \cite{LaGuardia:2013IV}
(w.r.t. classical convolutional codes), in the sense that here we
construct arbitrary $\alpha$-constacyclic code, whereas in
\cite{LaGuardia:2013IV}, only negacyclic codes (i.e., $\alpha = -1$)
are generated. In the same context, Theorem~\ref{mainIIIA} is a
generalization of the results shown in \cite{LaGuardia:2013II}.

\item  The families of almost MDS convolutional codes
constructed in Theorems~\ref{mainV} and \ref{mainVI} are new.
\end{itemize}

To illustrate the results, we present a table containing the
parameters of some new convolutional codes constructed in this
paper. The meaning of the parameters is clear from the context.

\begin{table}[!hpt]
\begin{center}
\caption{Convolutional MDS codes\label{table1}}
\begin{tabular}{|c |}
\hline Parameters of the new codes\\
\hline $(n, n-2i+1, 2; 1, 2i+2{)}_{q}$, $q$ odd prime\\ power,
$q-1=rk$, $k$ even,\\
$n=q+1$, $2\leq i \leq n/2 -2$\\
\hline ${(10, 7, 2; 1, 6)}_{9}$, ($r=4$)\\
\hline ${(10, 5, 2; 1, 8)}_{9}$, ($r=4$)\\
\hline ${(12, 9, 2; 1, 6)}_{11}$, ($r=5$)\\
\hline ${(12, 7, 2; 1, 8)}_{11}$, ($r=5$)\\
\hline ${(12, 5, 2; 1, 10)}_{11}$, ($r=5$)\\
\hline ${(26, 23, 2; 1, 6)}_{25}$, ($r=6$)\\
\hline ${(26, 17, 2; 1, 12)}_{25}$, ($r=6$)\\
\hline ${(26, 7, 2; 1, 22)}_{25}$, ($r=6$)\\
\hline Parameters of the new codes\\
\hline $(n, n-2i, 2; 1, 2i+3{)}_{q}$, $q$ odd prime\\ power,
$q-1=rk$, $k$ odd,\\ $n=q+1$, $2\leq i \leq n/2 -2$\\
\hline ${(12, 8, 2; 1, 7)}_{11}$, ($k=5$)\\
\hline ${(12, 6, 2; 1, 9)}_{11}$, ($k=5$)\\
\hline ${(12, 4, 2; 1, 11)}_{11}$, ($k=5$)\\
\hline ${(20, 16, 2; 1, 7)}_{19}$, ($k=9$)\\
\hline ${(20, 14, 2; 1, 9)}_{19}$, ($k=9$)\\
\hline ${(20, 12, 2; 1, 11)}_{19}$, ($k=9$)\\
\hline ${(20, 10, 2; 1, 13)}_{19}$, ($k=9$)\\
\hline ${(20, 4, 2; 1, 19)}_{19}$, ($k=9$)\\
\hline Parameters of the new codes\\
\hline $(n, n-2i+1, 2; 1, 2i+2{)}_{q}$, $q$ odd prime\\ power,
$q\equiv 1$ (mod $4$), $q-1=rk$, $k$ even,\\ $n=(q+1)/2$, $2\leq i \leq \frac{(n-1)}{2} -1$\\
\hline ${(7, 4, 2; 1, 6)}_{13}$, ($k=2$)\\
\hline ${(9, 6, 2; 1, 6)}_{17}$, ($k=9$)\\
\hline ${(9, 4, 2; 1, 8)}_{17}$, ($k=9$)\\
\hline ${(15, 12, 2; 1, 6)}_{29}$, ($k=4$; $r=7$)\\
\hline ${(15, 10, 2; 1, 8)}_{29}$, ($k=4$; $r=7$)\\
\hline ${(15, 8, 2; 1, 10)}_{29}$, ($k=4$; $r=7$)\\
\hline ${(15, 6, 2; 1, 12)}_{29}$, ($k=4$; $r=7$)\\
\hline ${(15, 4, 2; 1, 14)}_{29}$, ($k=4$; $r=7$)\\
\hline
\end{tabular}
\end{center}
\end{table}


\section{Asymmetric quantum codes}\label{secV}

Asymmetric quantum error-correcting codes (AQECC)
\cite{Steane:1996,Ioffe:2007,Sarvepalli:2009,LaGuardia:2013III} are
quantum codes defined over quantum channels where the probability of
occurrence of qudit-flip errors and phase-shift errors may be
different. The combined amplitude damping and dephasing channel
(specific to binary systems) is an example for a quantum channel
where these probabilities are different.

The scenario is the Hilbert space ${\cal H} = {\mathbb C}^{q^n} =
{\mathbb{C}}^{q} \otimes \ldots \otimes {\mathbb{C}}^{q}$ with a
orthonormal basis consisting of vectors $\mid$$x \rangle$, where the
labels $x$ are elements of ${\mathbb F}_{q}$. The unitary operators
$X(a)$ and $Z(b)$ ($a, b \in {\mathbb F}_{q}$) on ${\mathbb{C}}^{q}$
are defined by $X(a)$$\mid$$x \rangle =$$\mid$$x + a\rangle$ and
$Z(b)$$\mid$$x \rangle = w^{tr_{q/p}(bx)}$$\mid$$x\rangle$,
respectively, where $w=\exp (2\pi i/ p)$ is a $p$-th root of unity.
Assume that ${\bf a}= (a_1, \ldots , a_n)$ and ${\bf b}= (b_1,
\ldots , b_n)$ are vectors in ${\mathbb F}_{q}^{n}$ and consider
$X({\bf a})= X(a_1)\otimes \ldots \otimes X(a_n)$ and $Z({\bf b})=
Z(b_1)\otimes \ldots \otimes Z(b_n)$ be the tensor products of $n$
error operators. The set ${\bf E}_{n} = \{ X({\bf a})Z({\bf b}) \mid
{\bf a}, {\bf b} \in {\mathbb F}_{q}^{n} \}$ is an \emph{error
basis} on the complex vector space ${\mathbb C}^{q^n}$ and the set $
{\bf G}_{n} = \{w^c X({\bf a})Z({\bf b}) \mid {\bf a}, {\bf b} \in
{\mathbb F}_{q}^{n} , c\in {\mathbb F}_p \}$ is the \emph{error
group} associated with ${\bf E}_{n}$. For a quantum error $e = w^c
X({\bf a})Z({\bf b}) \in {\bf G}_{n}$ the $X$-weight is given by
wt$_{X}(e) = \ \# \{i: 1\leq i\leq n | a_i \neq 0\}$; the $Z$-weight
is defined as wt$_{Z}(e) = \ \# \{i: 1\leq i\leq n | b_i \neq 0\}$
and the symplectic (or quantum) weight swt$(e) = \ \# \{i: 1\leq
i\leq n | (a_i , b_i)\neq (0, 0)\}$. An AQECC with parameters ${((n,
K,d_{z}/ d_{x}))}_{q}$ is an $K$-dimensional subspace of the Hilbert
space ${\mathbb C}^{q^n}$ and corrects all qudit-flip errors up to
$\lfloor \frac{d_{x}-1}{2} \rfloor$ and all phase-shift errors up to
$\lfloor \frac{d_{z}-1}{2} \rfloor$. An ${((n, q^{k}, d_{z}/
d_{x}))}_{q}$ code is denoted by ${[[n, k, d_{z}/ d_{x}]]}_{q}$.
Next we recall the well-known CSS construction:

\begin{lemma}\cite{Ketkar:2006,Calderbank:1998,Nielsen:2000}(CSS
construction)\label{CSS1} Let $C_1$ and $C_2$ denote two classical
linear codes with parameters ${[n, k_1, d_1]}_{q}$ and ${[n, k_2,
d_2]}_{q}$, respectively. Assume that $C_2^{\perp}\subset C_1$. Then
there exists an AQECC with parameters $[[n, k_1+k_2-n,
d_{x}/d_{z}]{]}_{q}$, where $d_{x}=$wt$(C_2 \backslash C_1^{\perp})
\}$ and $d_{z}=$wt$(C_1 \backslash C_2^{\perp})$. The resulting
AQECC code is said pure if, in the above construction, $d_{x} =$
wt$(C_2)$ and $d_{z} =$ wt$(C_1)$.
\end{lemma}

The following result establishes the Singleton bound to AQECC, i.e.,
the asymmetric quantum Singleton bound (AQSB):

\begin{lemma}\cite[Lemma 3.2]{Sarvepalli:2009}
A pure asymmetric $[[n, k, d_{x}/d_{z}]]_{q}$ CSS code satisfies $k
\leq n -d_{x}-d_{z} +2$.
\end{lemma}
If an AQECC $Q$ have parameters attaining the AQSB with equality we
say that $Q$ is maximum-distance-separable (MDS) code. In the
following results, new MDS AQECC are constructed.

\begin{theorem}\label{mainasyI}
Assume all the hypotheses of Lemma~\ref{cyclo1} hold. Then there
exist asymmetric quantum MDS codes with parameters $[[n, 2(j-i),
(n-2j) / (2i+2)]{]}_{q}$, for every $0\leq i\leq j \leq n/2 -2$.
\end{theorem}
\begin{proof}
Let $C_2^{\perp}$ be the ($\alpha$-constacyclic) code with defining
set ${ \mathcal Z} = \bigcup_{l=0}^{j}{\mathbb C}_{(s+rl)}$, where
${\mathbb C}_{s}$ and ${\mathbb C}_{(s-rl)}$ are given in
Lemma~\ref{cyclo1}. Then $C_2^{\perp}$ has parameters $[n, n-2j-1,
2j+2{]}_{q}$. Next, consider $C_1$ be the ($\alpha$-constacyclic)
code with defining set ${ \mathcal Z} = \bigcup_{l=0}^{i}{\mathbb
C}_{(s+rl)}$, where $0\leq i\leq j\leq n/2 -2$; $C_1$ has parameters
$[n, n-2i-1, 2i+2{]}_{q}$. From construction one has $C_2^{\perp}
\subset C_1$; since $C_2^{\perp}$ is a MDS code then also is its
(Euclidean) dual code $C_2$, with parameters $[n, 2j+1,
n-2j{]}_{q}$. The dimension of the resulting CSS codes equals
$2(j-i)$. Applying the CSS construction, we obtain an AQECC with
parameters $[[n, 2(j-i), d_{x}\geq (n-2j)/ d_{z}\geq
(2i+2)]{]}_{q}$. Finally, from the AQSB, we have an $[[n, 2(j-i),
(n-2j)/ (2i+2)]{]}_{q}$ AQECC for every $0\leq i\leq j\leq n/2 -2$.
Note that the parameters of these codes attain the AQSB with
equality, hence they are MDS codes.
\end{proof}

\begin{theorem}\label{mainasyII}
Assume all the hypotheses of Lemma~\ref{cyclo2} hold. Then there
exist asymmetric quantum MDS codes with parameters $[[n, 2(j-i),
(n-2j-1)/ (2i+3)]{]}_{q}$ for every $0\leq i\leq j\leq n/2 -2$.
\end{theorem}
\begin{proof}
Let $C_2^{\perp}$ be the code with defining set ${ \mathcal Z} =
\bigcup_{l=0}^{j}{\mathbb C}_{(t+rl)}$, where ${\mathbb C}_{t}$ and
${\mathbb C}_{(t+rl)}$ are given in Lemma~\ref{cyclo2}, and assume
that $C_1$ is the code with defining set ${ \mathcal Z} =
\bigcup_{l=0}^{i}{\mathbb C}_{(t+rl)}$, with $0\leq i\leq j\leq n/2
-2$. Proceeding similarly as in the proof of Theorem~\ref{mainasyI},
the result follows.
\end{proof}

\begin{theorem}\label{mainasyIII}
Assume all the hypotheses of Lemma~\ref{cyclo3} hold. Then there
exist asymmetric quantum MDS codes with parameters $[[n, 2(j-i),
(n-2j) / (2i+2)]{]}_{q}$ for every $0\leq i\leq j\leq
\frac{(n-1)}{2} -1$.
\end{theorem}
\begin{proof}
Similar to that of Theorem~\ref{mainasyI}.
\end{proof}

\begin{remark}
It is important to observe that one can deduce an analogous of
Theorem~\ref{mainIV} for asymmetric quantum codes, generating in
this way more optimal AQECC.
\end{remark}

\begin{theorem}\label{mainasyIV}
Assume all the hypotheses of Theorem~\ref{mainclasIIIA} hold, with
$t \geq 3$. Then there exist asymmetric quantum MDS codes with
parameters $[[n, 2(j-i), (n-2j-1) / (2i+3)]{]}_{q}$ for every $0\leq
i\leq j\leq \frac{(n-1)}{2} -2$.
\end{theorem}
\begin{proof}
Similar to that of Theorem~\ref{mainasyI}.
\end{proof}

\begin{remark}
It is interesting to observe that if all hypotheses of
Theorem~\ref{mainclasIIIB} are true, more asymmetric quantum MDS
codes can be constructed.
\end{remark}

\begin{remark}
Although the parameters of our asymmetric quantum MDS codes are not
new (see \cite{Ezerman:2013}) the construction proposed here differs
from that given in \cite{Ezerman:2013}: those codes of length
$n=q+1$ are derived from (classical) extended generalized
Reed-Solomon (GRS) codes and those of length $n=(q+1)/2$ are derived
from (classical) GRS codes, whereas our AQECC are derived from
classical constacyclic codes. Therefore, our construction is an
alternative way to produce asymmetric quantum MDS codes of lengths
$n=q+1$ and $n=(q+1)/2$.
\end{remark}

\section{Summary}\label{secVI}
In this paper we have constructed several new families of optimal
convolutional codes derived from classical constacyclic codes.
Additionally, we also have constructed new families of almost MDS
block and convolutional codes. Moreover, based on an alternative
procedure, we have constructed families of asymmetric quantum codes
known in the literature. These results show that the class of
constacyclic codes is also a good resource in order to find codes
with good or even optimal parameters.

\section*{Acknowledgment}
I would like to thank the anonymous referee for his/her valuable
comments and suggestions that improve significantly the quality and
the presentation of this paper. This research has been partially
supported by the Brazilian Agencies CAPES and CNPq.

%

\end{document}